\documentclass[acmsmall,screen,nonacm,review=false,timestamp=false]{acmart}

\usepackage{microtype}
\usepackage{graphicx}
\usepackage{subfigure}
\usepackage{booktabs} %
\usepackage{algorithm}
\usepackage{algorithmic}
\usepackage{amsmath}

\DeclareMathOperator*{\argmin}{arg\,min}
\usepackage{hyperref}
\usepackage{lipsum}
\usepackage{xcolor}

\usepackage{mathtools}
\usepackage{amsthm}

\usepackage[capitalize,noabbrev]{cleveref}

\theoremstyle{plain}
\newtheorem{theorem}{Theorem}[section]

\newtheorem{lemma}[theorem]{Lemma}
\newtheorem{corollary}[theorem]{Corollary}
\theoremstyle{definition}

\newtheorem{assumption}[theorem]{Assumption}
\theoremstyle{remark}

\AtBeginDocument{%
  }

\setcopyright{acmlicensed}
\copyrightyear{2018}
\acmYear{2018}
\acmDOI{XXXXXXX.XXXXXXX}
\acmConference[Conference acronym 'XX]{Make sure to enter the correct
  conference title from your rights confirmation email}{June 03--05,
  2018}{Woodstock, NY}
\acmISBN{978-1-4503-XXXX-X/2018/06}

\begin{document}

\title{Counterfactual simulations for large scale systems with burnout variables}

\author{Benjamin Heymann}

\email{b.heymann@criteo.com}
\orcid{0000-0002-0318-5333}
\affiliation{%
  \institution{Joint team Fairplay, ENSAE, and Criteo AI LAB}
  \city{Paris}
  \country{France}
}

\renewcommand{\shortauthors}{Heymann}

\begin{abstract}
We consider large-scale systems influenced by burnout variables—state variables that start active, shape dynamics, and irreversibly deactivate once certain conditions are met.
Simulating “what-if” scenarios in such systems is computationally demanding, as alternative trajectories often require sequential processing, which does not scale very well. This challenge arises in settings like online advertising, because of campaigns budgets, complicating counterfactual analysis despite rich data availability.

We introduce a new type of algorithms based on what we refer to as  uncertainty relaxation, that enables efficient parallel computation, significantly improving scalability for counterfactual estimation in systems with burnout variables.
\end{abstract}

\maketitle

\section{Introduction}
 This paper  introduces scalable estimation methods for large-scale systems influenced by burnout variables, that is, state variables that start active, shape dynamics, and irreversibly deactivate once certain conditions are met. 
For expository purpose, and because it was our initial motivation, 
we will be  focusing on  applications for the field of digital advertising.

So the setting is as follows: retail media and search advertising platforms enable advertisers to display ads to users through real-time auctions, where each opportunity to show an ad is auctioned as it arises. Advertisers participate in these auctions via campaigns, which define parameters such as target audience, budget, and cost-per-click (CPC). The platform uses these inputs to determine allocations and payments, often leveraging machine learning predictions—such as estimated click probabilities—to optimize outcomes. The specific allocations and payments also depend on the platform design, such as the choice of auction rule and bid multipliers.  

A key challenge in this setting is budget constraints, which impose a cap on campaign spending and pause bidding once the budget is depleted. These constraints create budget coupling effects, where a change in one campaign’s bidding behavior can indirectly influence other campaigns by altering competition dynamics. This interdependence complicates counterfactual estimation, as modifying a single campaign’s strategy can trigger cascading effects across the platform.

While a sequential replay of past auctions could, in principle, enables the 
simulation of alternative platform configurations, this approach does not scale, and we need alternatives leveraging parallel computing frameworks-- such as MapReduce-- that offers efficient way to process large-scale auction data. Our work focuses on developing a scalable approach to counterfactual estimation that accounts for budget constraints while maintaining computational efficiency.

\paragraph{Advertising platforms}

Advertising platforms play a fundamental role in the global economy. They serve as the foundation for major corporate revenue models, as seen in companies like Google, Microsoft, Amazon, and Meta. These platforms connect advertisers with content providers—whether search engines, newspapers, retailers, or social networks—and concentrate significant economic and political influence.
Understanding the impact of design choices in these platforms is therefore critical, as such decisions can have far-reaching economic, business, and societal consequences. While substantial research has focused on optimizing platform performance, relatively little is known about how to leverage real data to evaluate  alternative design or parameter changes, such as modifications to cost-per-click (CPC) pricing or auction rules. This stands in contrast to real-time bidding (RTB)~\cite{choi2020online}, where a well-developed causal inference literature exists. However, counterfactual estimation on advertising platforms presents  distinct challenges: on the one hand, they have access to significantly more information than individual bidders, but on the other hand they need to account for more complex dynamics because they control the full system.

\paragraph{Search}
A search advertising platform is a system that facilitates the auction-based placement of ads within search engine results. When a user enters a query, the platform runs an auction in real time to determine which ads to display and in what order. Advertisers bid on keywords relevant to their products or services, and their ad placements depend on factors such as bid amounts, ad relevance, and expected user engagement. The most common pricing model is cost-per-click (CPC), where advertisers pay only when users click on their ads. Leading search ad platforms, such as Google Ads and Microsoft Advertising, optimize these auctions to maximize revenue while balancing user experience and advertiser value. These platforms play a crucial role in digital marketing, driving targeted traffic to businesses while generating significant revenue for search engines.
\paragraph{Retail media}

Retail media is the digital evolution of traditional in-store advertising, where brands have long sought to influence shoppers at the moment of purchase. Classic examples include in-store promotions, product demonstrations, and free samples—like offering cheese tastings in supermarkets—to capture consumer attention and drive immediate sales. While these strategies relied on physical store space, modern retail media brings the same concept online, where digital marketplaces like Amazon and Target serve as the new storefronts.

Today, retail media is a multi-billion dollar industry where advertisers compete for ad placements through real-time auctions. Brands, agencies, and marketers set bid prices and budgets to display their ads on key retail platform locations, such as product pages, search results, and category pages. When a user visits the platform, an auction determines which ads are shown, often using first-price or second-price auction models. Once an advertiser's budget is depleted, their campaign is paused until funds are replenished.

For marketplace operators, designing these auction systems is a complex challenge with major business implications. Decisions about auction rules, pricing models, and ad placements impact revenue, advertiser satisfaction, and overall platform efficiency. Counterfactual estimation—simulating how different auction designs would perform—helps operators evaluate potential changes without real-world risks. However, running full historical replays of past auctions is not computationally reasonable most of the time, making efficient estimation methods essential for optimizing retail media platforms.

\paragraph{Contributions}
Our main contribution is the analysis of a parallelizable algorithm for counterfactual estimation in advertising platforms. While the algorithm itself can be executed in parallel, a naive approach requires multiple calls to obtain reliable estimates, which can be computationally expensive. To address this, we mention a strategy that reduces the number of required executions, enabling the algorithm to be run only once.

For the analysis of the method, we introduce structural assumptions on the data-generating process that are both realistic and analytically useful. We then provide a theoretical analysis of the algorithm’s performance, establishing concentration guarantees for its estimates. A key insight from our results is that the primary challenge lies in accurately estimating capping-out times (i.e., when budgets are depleted).

To tackle this issue, we introduce \textsc{sort2aggregate} for estimating capping-out times efficiently. Our experimental results show that \textsc{sort2aggregate} significantly improves estimation accuracy while maintaining computational efficiency.

\section{Related work}
\paragraph{Budgets and Capping out times}
Advertising campaigns typically have budget caps~\cite{abrams2008ad,aggarwal2024auto,balseiro2020dual,mehta2007adwords,10.1145/3589335.3648331}, meaning a campaign is paused once its spending exceeds the allocated budget. While this constraint ensures advertisers do not overspend, it introduces a challenge for counterfactual estimation. Since all campaigns on the platform interact dynamically, a change in bidding behavior for one campaign can indirectly affect others through budget constraints. This coupling makes it difficult to predict the isolated impact of a design change, as altering one campaign’s strategy may shift competition, reallocating impressions and accelerating or delaying when other campaigns exhaust their budgets.

\paragraph{Counterfactual estimation}
Counterfactual estimation~\cite{peters2017elements}, rooted in the foundational works of researchers such as Pearl~\cite{pearl1995causal,pearl2009causality}, Neyman and Rubin~\cite{neyman1923application,imbensCausalInferenceStatistics2015}, is a key concept in causal inference and has become a central topic in artificial intelligence. In the context of advertising marketplaces, counterfactual analysis helps quantify the trade-offs between different stakeholders' objectives, such as revenue maximization for the platform, return on investment for advertisers,  and user experience. 

A direct way to estimate causal effects is through A/B testing or randomized techniques such as Inverse Propensity Scoring (IPS). However, relying solely on these methods is limiting. Running multiple A/B tests on the same scope is typically infeasible, as each test requires an adequate duration to achieve sufficient statistical power. Moreover, setting up A/B tests introduces potential pitfalls: human errors in experiment design can lead to misleading results or even disrupt platform stability. Additionally, A/B testing relies on specific statistical assumptions—such as independence between experimental units—that may not hold in complex advertising environments. For instance, evaluating the effect of modifying bid multipliers in an auction is challenging, as such changes can indirectly influence market dynamics and competitors behavior. Similarly, offline policy estimation methods suppose the  candidate policy to be tested not too far from the production policy, which might prevent answering counterfactual questions on macroscopic  changes on the platform.

To overcome these limitations, causal inference methods typically leverage structural assumptions about the data-generating process, allowing for more robust conclusions beyond simple observational correlations. 
Developing accurate and scalable counterfactual models remains a critical challenge, as advertising platforms must balance computational feasibility with statistical reliability.

\paragraph{Causal method in programmatic advertising}

Causal inference methods play a critical role in optimizing advertising strategies~\cite{bottou2013counterfactual}, in particular  in real-time bidding (RTB) marketplaces, where advertisers must assess the impact of repeated user interactions. \cite{bompaire2024fixed,10.1145/3447548.3467280,moriwaki2021real} highlight the importance of accounting for the diminishing marginal effect of repeated ad displays, proposing a causal model to refine bidding strategies by valuing each display individually. 
Supposing the causal effects can be measured, \cite{heymann2024repeated} show the complexities of repeated bidding with causal effects on the user timeline using an optimal control framework. 
A view complementary to ours is taken in~\cite{betlei2024maximizing}, that address the challenge of maximizing policy allocation success in dynamic online advertising systems, reinforcing the need for precise causal estimation. 
Close to our spirit is the work of~\cite{liao2023statistical} that leverages the assumption that a pacing equilibrium is reached in a first price auction.  Under this assumption in our setting, all the campaigns cap out at the end, which greatly simplify  Algorithm~\ref{alg:example}.
Last,  \cite{waisman2024online} develop an online learning perspective.

\paragraph{Computation at scale}%
Massive amounts of data are handled by modern advertising platforms, which require robust frameworks and distributed computing infrastructure so that information can be processed and analyzed efficiently. 
MapReduce, a programming model designed for processing distributed data efficiently, is one of the fundamental paradigms enabling large-scale computation \cite{10.1145/1327452.1327492}. MapReduce breaks down large computations into two key steps: the map step, where data is partitioned and processed in parallel across multiple machines, and the reduce step, where intermediate results are aggregated into a final output. This approach enables frameworks such as Apache Spark and Hadoop to carry out parallel computations over datasets that are too voluminous to be accommodated on a solitary machine, thereby effecting a substantial enhancement in processing speed and scalability.

\section{Model of large scale systems with burnout
variables}
Next,  we introduce an abstraction that could represent the sequence of auction data points of an advertising platform. We use this application to ground this abstraction into a real life scenario. 
We denote by $\mathcal{E}$ a finite set of  data points of size   $N\in \mathbb{N}$ that corresponds to events such as auction. 
We start with the following structuring assumption: 
\begin{assumption}[Random order relaxation]
\label{assumption:random_order}
The realized sequence of auctions, denoted by $(\boldsymbol{e_1}\ldots\boldsymbol{e_N})$  is a random sample without replacement from $\mathcal{E}$.
\end{assumption}
We denote by $\mathcal{F}_n$ the filtration generated by the process $(\boldsymbol{e}_{n})_{n\in[1,N]}$.
The state $s$ of the platform is the level of spend of its campaigns, that is, if $c$ is the index of a campaign in the finite set of campaigns $\mathcal{C}$, and $n\in[1:N]$ is the index of an auction event, then $s^c_n$ is the cumulated  spend of the campaign $c$ after the resolution of the $n$th auction.
If $b=(b^c)_{c\in\mathcal{C}}$ is the vector of  \textit{budgets}, that maps each campaign $c$ to its budget $b^c>0$,  $a_n$ refers to the binary \textit{activation} vector in $\{0,1\}^\mathcal{C}$ defined by, for $c\in\mathcal{C}$,
\begin{align}
    a^c_n =
    \begin{cases}
        1, & \text{if } s_n^c<b \\
        0, & \text{otherwise},
    \end{cases}
\end{align}
 and $f:\mathcal{E}\times \{0,1\}^\mathcal{C}\to\mathbb{R}_{+}^\mathcal{C}$ is a map that, given the event $e\in \mathcal{E}$ and the list of campaigns that are active (encoded   by a binary vector from $\{0,1\}^\mathcal{C}$), outputs the increment of spend of each campaign~\footnote{here we focus on the spend, but the dynamics could encode other quantities of interest}.
We suppose the spend dynamics follow the following rule: 
\begin{align}
    s_0^c &= 0  \quad \forall c\in \mathcal{C},\\
    s_{n}   &= s_{n-1}+ f(\boldsymbol{e}_n ,a_{n-1})\quad \forall n\in [1:N].
\end{align}
 Hence, $f$ can be interpreted as the vector of spending "speed" of the campaigns.
 When a campaign $c\in\mathcal{C}$ is not active, the spending is stopped: $a^c = 0\implies f^c(.,a) = 0 $.
We next state two structuring assumptions on the auction rule $f$.

\begin{assumption}[small individual contribution]
\label{assumption:small-individual-contribution}
    There exists $C>0$ such that for any $e\in\mathcal{E}$, $a\in \{0,1\}^\mathcal{C}$ and $c\in\mathcal{C}$ 
    \begin{align}
        f^c(e,a)< \frac{C}{N}.
    \end{align}
\end{assumption}

Assumption~\ref{assumption:small-individual-contribution}, is reminiscent from the \textit{ad words} literature~\cite{mehta2007adwords}, and implies that the contribution of any individual auction to the campaign spend is small compared to the budgets.

We introduce the following shorthand: for any binary vector $a\in\{0,1\}^\mathcal{C}$ and $c\in\mathcal{C}$ we denote by 
 $a-\{c\}$  the binary vector $\tilde{a}$ such that 
\begin{align}
    \tilde{a}^{c'} =
    \begin{cases}
        a^{c'}, & \text{if } c' \neq c \\
        0, & \text{otherwise}.
    \end{cases}
\end{align}
\begin{assumption}[$(\gamma,\delta,\epsilon)$-smoothness]
\label{assumption:smooth}
There exists $\gamma>0$, $\delta>0$, $\epsilon>0$, such that with a probability greater than $1-\delta$, 
\begin{align}
\label{eq:gamma}
    \sum_{i=m}^n\left[ f^{c'}(\boldsymbol{e}_{i},a-\{c\})-f^{c'}(\boldsymbol{e}_{i},a)\right] \leq \gamma \sum_{i=m}^n[f^{c}(\boldsymbol{e}_{i},a)]+\epsilon,
\end{align}
for any $(c',c)\in\mathcal{C}$, $(n,m)\in[1:N]$, and $a\in [0,1]^\mathcal{C}$.

\end{assumption}
The interpretation of condition~\eqref{eq:gamma} is that the deactivation of a campaign $c$ should not, on average impact any other campaign more than a factor  of  the "speed" of campaign $c$ before deactivation.

In this formulation, the budget may exceed its limit by a small margin, but Assumption~\ref{assumption:small-individual-contribution} ensures that this excess remains at most \( \frac{C}{N} \). The smoothness assumption holds whenever \( \gamma \) is sufficiently large. Typically, in a first-price auction, we have \( \gamma \leq 1 \). However, in settings where multiple campaigns coexist without a dominant one, we argue that \( \gamma \) should scale as \( O\left(\frac{1}{|\mathcal{C}|}\right) \). In a second-price auction, the smoothness assumption rules out cases where a large campaign is disproportionately influenced by a much smaller one. The random order assumption facilitates statistical guarantees on the algorithm’s counterfactual statements, it has been used in the ad words literature to provide competitive ratio on allocation strategies~\cite{kesselheim2014primal,nikhilr.devanurAdwordsProblemOnline2009}. Relevant marketplaces include first-price auctions with or without pacing mechanisms, such as bid scaling and random throttling.

\section{Sequential simulation}

Mostly, $f$ encodes the auction rules of the platform, but it may also include ML inferences that influence the allocation decision. This paper is based on the observation that if, for a hypothetical $\tilde{f}$, given any $e \in \mathcal{E}$ and any $a \in \{0,1\}^{\mathcal{C}}$, we can compute $\tilde{f}(e,a)$, then we should also be able to compute what would have happened had the platform switched from auction rule $f$ to $\tilde{f}$.

The counterfactual estimation relies on the assumption that $\mathcal{E}$ captures all auction-relevant state variables. However, this is an approximation. In practice, some dependencies may be omitted, such as the fact that an advertiser’s valuation for an opportunity may depend on how many banners they have already shown to the user. This assumption effectively rules out long-term dependencies between auction events, which are often neglected in the literature, aside from a few exceptions like \cite{heymann2024repeated,heymann2024pragmatic}.

Additionally, when only the auction rule changes while the ML model remains fixed, $\tilde{f}$ should use the same ML models as those deployed online. This ensures consistency with the observed auction environment. However, if these ML models were trained on outcomes influenced by $f$, using them under $\tilde{f}$ could introduce biases. Our approach assumes that these biases remain negligible.

Since the rule that was used does not affect our result—provided we have access to $\mathcal{E}$—we suppose, without loss of generality, that we are interested in estimating the counterfactual state $s$, which differs from the historical realization when computed under a different $f$. Because $s$ results from a sequential computation, we refer to it as the \textit{sequential simulation}. We will show how this sequential simulation can be efficiently estimated using a \textit{parallel simulation}, enabling large-scale counterfactual analysis on a distributed computing system.

This paper contributes  an algorithm that unlocks such parallelism. 
Here is a hint of the algorithm philosophy. 
Take $T$  a large non-negative integer that represents a time horizon,  and $x_t>0$ for $t\in [T]$ a sequence of positive numbers that can be interpreted as prices. 
Define  $S_T$ sequentially by 
\begin{align*}
    S_1 &= 0 \\
    S_{t+1}&=\min(S_t+x_t,B)
\end{align*}
The definition of $S_T$ implies it can be computed sequentially by first computing $S_1$, then $S_2$, and, keeping on until one reaches $S_T$.
But sequential computations do not scale very well.  The following trivial algorithm circumvents this limit and allows for  parallelization.
\begin{algorithm}[h]
\caption{Trivial Algorithm}
\label{algo:0}
\begin{algorithmic}[1]
\STATE \textbf{Input:} $B$, $(x_t)_{t\in [T]}$
\STATE \textbf{Output:} $\min(B, \sum_{t=1}^{T} x_t)$
\STATE \textbf{return} $\min(B, \sum_{t=1}^{T} x_t)$
\end{algorithmic}
\end{algorithm}
\\
 Algorithm~\ref{algo:0} does not require the sum of the $x_t$ to be made in a specific order.\textbf{ Hence, the summation can be distributed on a cluster of computing machines.}
 We extend this idea to a more general class of system.

\section{Parallel simulation}
The \textbf{Parallel Simulation Algorithm} displayed in Algorithm~\ref{alg:example} aims to efficiently estimate counterfactual auction outcomes by simulating the allocation process in a structured and scalable manner. The algorithm takes as input a set of campaigns $\mathcal{C}$, their corresponding budgets $\{b^c\}_{c \in \mathcal{C}}$, an ordered sequence of auction events $(\boldsymbol{e_1},\ldots,\boldsymbol{e_N})$, and a function $ f(\mathbf{e}, a) $ that encodes the auction mechanism. The algorithm iteratively processes auction events while maintaining budget constraints and tracking campaign allocations in a way that allows parallelization.  

At initialization, the counters for each campaign are set to zero, and all campaigns are assumed to be active. At each iteration, the algorithm estimates the expected value $ F_{i+1} $ of the auction function given the next event in sequence, conditioned on the history of previously observed events. Based on this estimation, the campaign $\hat{c}_{i+1}$ that is expected to deplete its budget the fastest is identified. The next event index $\hat{N}_{i+1}$ is then determined by computing the  number of events that can be processed before the budget of the selected campaign runs out, ensuring that it does not exceed the total number of events $ N $.  

Within this range of events, the algorithm assigns actions $\hat{a}_n$ according to the current active set $\hat{A}_i$, and updates the sequence of accumulated spending $\hat{s}_n$ using the auction function $ f(\boldsymbol{e}_n, \hat{A}_i) $. Once all events in this range have been processed, the campaign $\hat{c}_{i+1}$ is removed from the active set $\hat{A}_i$, indicating that it has exhausted its budget. The iteration index is then incremented, and the loop repeats with the remaining active campaigns.  

This approach enables each iteration of the loop to be executed at scale. However, when the number of campaigns is large, the required number of iterations becomes prohibitively high, limiting the practicality of this algorithm. 

\begin{algorithm}[ht]
\caption{Parallel simulation}
\label{alg:example}
\begin{algorithmic}[1]
\REQUIRE A set of campaigns $\mathcal{C}$, budgets $\{b^c\}_{c \in \mathcal{C}}$, an ordered sequence of events $(\boldsymbol{e_1},\ldots,\boldsymbol{e_N})$, and a function $f(\mathbf{e}, a)$.
\ENSURE Updated sequences $\{\hat{s}_n\}$ and $\{\hat{a}_i\}$.

\STATE \textbf{Initialization:}
\STATE \quad $\hat{s}_0^c \leftarrow 0,\quad \hat{a}_0^c \leftarrow 1,\quad \forall c \in \mathcal{C}$
\STATE \quad $\hat{N}_0 \leftarrow 0, i \leftarrow 0$

\WHILE{$\hat{N}_i<N$ and $\hat{A}_i \neq 0 $}
    \STATE $F_{i+1} \leftarrow \mathbb{E}\bigl[f(\boldsymbol{e}_{\hat{N}_i + 1}, \hat{a}_i) \,\big|\,
    \{\hat{N}_i, \boldsymbol{e}_1, \ldots, \boldsymbol{e}_{\hat{N}_i}\}\bigr]$

    \STATE $\hat{c}_{i+1} \leftarrow \displaystyle \argmin_{c \in \mathcal{C} \,\text{ s.t. }\, \hat{a}_i^c=1}
    \Bigl(\frac{b^c - \hat{s}_{\hat{N}_i}^c}{F_{i+1}^c}\Bigr)$

    \STATE $\hat{N}_{i+1} \leftarrow \min\Bigl(\hat{N}_i \;+\; \bigl\lfloor\frac{b^{\hat{c}_{i+1}} - \hat{s}_{\hat{N}_i}^{\hat{c}_{i+1}}}{F_{i+1}^{\hat{c}_{i+1}}} \bigr\rfloor,\; N\Bigr)$

    \FOR{$n = \hat{N}_i + 1$ to $\hat{N}_{i+1}$}
        \STATE $\hat{a}_n \leftarrow \hat{A}_i$
        \STATE $\hat{s}_n \leftarrow \hat{s}_{n-1} + f(\boldsymbol{e}_n, \hat{A}_i)$
    \ENDFOR

    \STATE $\hat{A}_{i+1} \leftarrow \hat{A}_i - \{\hat{c}_{i+1}\}$
    \STATE $I\leftarrow i $
    \STATE $i \leftarrow i+1$
\ENDWHILE
\end{algorithmic}
\end{algorithm}

We first state the following concentration result, that is proved in the appendix. 
\begin{lemma}
\label{lemma:hoeffding}
Let $n\leqslant N$, $c\in\mathcal{C}$, $\alpha\in\{0,1\}^\mathcal{C}$, and $t>0$, set $F=  \mathbb{E}\bigl[f^c(\boldsymbol{e}_{1}, \alpha) \bigr]$ then 
\begin{align*}
    \mathbb{P}\left(\mid\sum_{i=1}^n f^c(\boldsymbol{e}_i,\alpha) - n F\mid \geqslant t \right)\leqslant 2\exp\left(-\frac{2N t^2}{C^2}\right).
\end{align*}
\end{lemma}
\paragraph{Main theoretical result}
We can now state our main result concerning the analysis of Algorithm~\ref{alg:example}.
\begin{theorem}
\label{th}
  For $t$ small enough,     
for any $c\in\mathcal{C}$, 
with probability at least $1 - \delta- 2K\exp\left(-\frac{2N t^2}{C^2}\right)$, 
\begin{align}
    \mid s_{N}^c-\hat{s}_{N}^c\mid\leq (1+\gamma)^{K}\left(\frac{C}{N}+t+\gamma\epsilon +\epsilon\right).
\end{align}
\end{theorem}
\paragraph{Proof sketch}
The proof relies on the observation that the error $\mid s_n^c - \hat{s}_n^n\mid$  is controlled at $n=\hat{N}^{c_i}$ for $c=c_i$ by the error made on the campaigned that already capped out. This comes from the smoothness assumption. Then a bound in derived by induction. 

\begin{corollary}
\label{corr}
Suppose $\gamma \leq \frac{D}{K}$. 
  For $t$ small enough,     
for any $c\in\mathcal{C}$, 
with probability at least $1 - \delta- 2K\exp\left(-\frac{2N t^2}{C^2}\right)$, 
\begin{align}
    \mid s_{N}^c-\hat{s}_{N}^c\mid\leq e^{D}\left(\frac{C}{N}+t+\gamma\epsilon +\epsilon\right).
\end{align}
\end{corollary}

\paragraph{Insights from Theorem~\ref{th}}
While Algorithm~\ref{alg:example} may have limited practical applicability due to the number of iterations required, it establishes a fundamental principle for developing more efficient computing strategies. In particular, knowing either (a) the exact times when campaigns deplete their budgets or (b) simply the order in which they do so can significantly simplify computation.  

For instance, if the depletion order is known, the average value of $ f() $ over the relevant activation vectors can be computed in a scalable manner. This, in turn, allows us to determine the depletion times and reconstruct the full history efficiently. However, when the depletion order is unknown, the number of required evaluations of the average value of $ f() $ increases from $ O(|\mathcal{C}|) $ to $ 2^{|\mathcal{C}|} $, leading to a significant computational burden.  
We leverage this observation in the next section.

\section{\textsc{sort2aggregate}}
We next leverage the insight from the previous section to discuss strategies that could be both computationally efficient and statistically sound. 
First, observe that in Algorithm~\ref{alg:example}, the intermediate estimation of the capping out of the campaigns are not reused.
When $f^c$ is decreasing in  all components of the activation vector (aside from $c$'s), then we can use arguments from lattice theory to build an algorithm that keeps track of those intermediate capping out times and that converges (by Tasky's fixed-point theorem, see~\cite{heymann2024reverse,topkis1998supermodularity} ). 
Such a strategy would however still require multiple steps of computation at scale. 
\paragraph{\textsc{Sort2aggregate} algorithm}
We introduce the times when the $i$th campaign exhausts its budget, formally 
\begin{align}
N_i = \min\{n\in[0,N]:\sum_{c\in\mathcal{C}}(1-a_n^c)\geqslant i\},
\end{align}
and $c_i$ be the campaign associated with the stopping times $N_i$.
The pseudo-code given in Algorithm ~\ref{alg:rank_estimation} provide a general strategy for scaling-up the counterfactual estimation. 
After the first step, the computations can be parallelized because the activation vectors have been identified. Notably, any errors in one step should become apparent in the subsequent step, providing a built-in safeguard for detecting inconsistencies.

Algorithm~\ref{alg:heuristic_frequency} provides a method to derive an estimate of the $N_i$s. It is based on replacing $N^c$ by $N \cdot \pi_c$ with $\pi_c\in [0,1]$. 
We now provide a theoretical justification for the convergence of Algorithm~\ref{alg:heuristic_frequency}  using the framework of variational inequalities (VIs).

Let $\pi \in [0,1]^{|\mathcal{C}|}$ represent the (scaled) finish time for each campaign (end of day of capping out), and define the residual operator
\begin{align}
   G(\pi) = F(\pi)-b, 
\end{align}
where $F(\pi)$ represents the expected cumulative spend as a function of $\pi$ 
and $b$ is the vector of budgets. 
The conditions on $\pi$ corresponds to the
complementarity conditions
\begin{align*}
        0 \le 1-\pi_c, \quad 0 \geq G_c(\pi), \quad
    \underbrace{(1-\pi_c)}_{\text{0 if c finishes the day}} \cdot \underbrace{G_c(\pi)}_{\text{0 if c  cap-out}} = 0,
    \quad \forall c \in \mathcal{C},
\end{align*}
which is equivalent to the variational inequality problem $\mathrm{VI}(K,G)$:
\[
    \text{Find } \pi^\star \in K \text{ such that } 
    \langle G(\pi^\star),\, \pi - \pi^\star\rangle \ge 0 
    \quad \forall\, \pi \in K,
\]
where $K=[0,1]^{|\mathcal{C}|}$.
Algorithm~\ref{alg:heuristic_frequency} can be interpreted as a projected fixed-point iteration on the VI:
\[
    x^{(t+1)} 
    = \Pi_{K}\Bigl( x^{(t)} - \eta \, G(x^{(t)}) \Bigr),
\]
where $\Pi_K$ is the Euclidean projection onto the feasible set $K$ 
and $\eta>0$ is a step size. 
Notably, this update depends only on the residual $G(x) = F(x)-B$ 
and does not require the Jacobian of $F$, making it scalable for 
large systems.
Under the assumptions that $F$ is monotone in each coordinate 
and Lipschitz continuous with constant $L$, the residual-only 
projected iteration is a (projected)  linearized Jacobi dynamics~\cite{harker1990finite}. Identification of conditions on the primitive of the problem that ensures convergence is left for further work.
Last, observe that Algorithm~\ref{alg:heuristic_frequency} itself can be slightly modified to be implemented at scale (stochastic gradient).
\begin{algorithm}[tb]
   \caption{\textsc{Sort2aggregate}: Method for counterfactual  estimation at scale }
   \label{alg:rank_estimation}
\begin{algorithmic}[1]
   \STATE \textbf{Step 1:} Get an estimate of the rank of the campaigns by $N_i$
   \STATE \textbf{Step 2:} Refine $N_i$ (optional)
   \STATE \textbf{Step 3:} Aggregate at scale
\end{algorithmic}
\end{algorithm}

\begin{algorithm}[tb]
   \caption{  $N_i$s Estimation}
   \label{alg:heuristic_frequency}
\begin{algorithmic}[1]
   \REQUIRE Number of events $N$, number of campaigns $|\mathcal{C}|$, budgets $\{b^c\}_{c \in \mathcal{C}}$, auction rule $f(\mathbf{e}, a)$, sampling rate $\rho$, optimization rate $\eta$, number of iterations $T$

   \STATE $k \leftarrow \text{round}(N \times \rho)$
   \STATE $\mathcal{E}^{\text{sampled}} \leftarrow$ sample $k$ events randomly from $[1, N]$ without replacement
   \STATE $\tilde{b}^c \leftarrow b^c / N, \quad \forall c \in \mathcal{C}$
   \STATE $\boldsymbol{\pi} \leftarrow \mathbf{1}$ (vector of ones of size $|\mathcal{C}|$)

   \FOR{$t = 1$ to $T$}
      \FOR{each event $e \in \mathcal{E}^{\text{sampled}}$}
         \STATE Draw $u \sim \text{Uniform}(0,1)$
         \STATE $a_c \leftarrow 1$ if $u_c < \pi_c$, else $0$, for each $c \in \mathcal{C}$ (activation vector)

         \STATE $s \leftarrow f(e, a)$ (spending per campaign)
         \STATE $\Delta \leftarrow (\tilde{b} - s)$
         \STATE $\boldsymbol{\pi} \leftarrow \max(\min(\boldsymbol{\pi} + \eta \cdot \Delta, 1), 0)$
      \ENDFOR
   \ENDFOR
\end{algorithmic}
\end{algorithm}
\paragraph{Computing time}
If $A$ is the wall clock time for solving one auction, the sequential simulation should take $N\cdot A$. 
By contrast, with sort2aggregate, we should expect $N\cdot A\cdot T\cdot \rho/N_{core}$ for the $N_i$ estimation, and then $N\cdot A/N_{core}$ for the aggregation.

\section{Numerical experiments}
\subsection{Fully synthetic data}
For reproducibility and simplicity purpose, we first test the ideas presented in this paper on generated data. The code will be made available in Julia.

The experimental setup simulates an auction-based advertising system where multiple campaigns compete for impressions in a first-price auction. Each auction event corresponds to an advertising opportunity, and each campaign assigns a valuation to an event based on a learned embedding representation. Specifically, each campaign $c$ and each event $e$ are represented as vectors in a shared embedding space of dimension $d$, where the embeddings are sampled from a multivariate normal distribution. The event embeddings are computed as:  

\begin{equation}
\mathbf{e}_i = \frac{\mathbf{e}_{\text{base}} + 3 \cdot \boldsymbol{\xi}_i}{4}, \quad \boldsymbol{\xi}_i \sim \mathcal{N}(0, I_d)
\end{equation}

where $\mathbf{e}_{\text{base}}$ is a global reference embedding and $\boldsymbol{\xi}_i$ represents random perturbations drawn from a standard normal distribution. Similarly, each campaign $c$ is associated with an embedding vector $\mathbf{r}_c$ drawn independently from $\mathcal{N}(0, I_d)$. The valuation of campaign $c$ for event $e_i$ is computed using a scaled inner product similarity measure:

\begin{equation}
 v_c(e_i) = \min(\frac{\exp\left(\frac{\mathbf{r}_c^\top \mathbf{e}_i}{2\sqrt{d}}\right)}{10},1)
\end{equation}

which ensures that valuations remain bounded while capturing campaign-event affinity. Given these computed values, an auction rule determines the winning campaign and its corresponding expenditure. The first-price auction mechanism is implemented such that the highest-bidding campaign wins and pays its bid, constrained by budget limitations.

The budgets of the campaigns are set in a structured manner to ensure that a significant portion of the campaigns deplete their budgets, creating a competitive auction environment. Specifically, the budgets are assigned as follows:

\begin{equation}
    b^c = k \cdot b_{\text{base}}, \quad \text{for } k \in \{1, 2, \dots, |\mathcal{C}|\}
\end{equation}

where $b_{\text{base}}$ is a fixed base value, and each campaign $c$ is assigned a budget proportional to its index $k$. This setup ensures that budgets scale linearly with campaign indices, leading to a mix of high- and low-budget campaigns.

Additionally, to create a realistic scenario where some campaigns saturate their budgets while others remain active, the base budget is chosen such that approximately 50\% of the campaigns reach their budget limits by the end of the auction process. This ensures a dynamic competitive landscape, making the auction outcomes more representative of real-world bidding environments.

We first check that the naive approach of sampling events from $\mathcal{E}$ and then do a sequential simulation with rescaled values that account for the sampling rate can be a bad idea:  in Figure~\ref{fig:icml-sampling}, we see how the predictions become irrelevant as the sampling rate increases.

In Figure~\ref{fig:icml-parallel}, we compare the output of the sequential and parallel simulations for the same environment. We can check that they are indeed extremely close.
In Figure~\ref{fig:frequenciesl} we observe the convergence of Algorithm~\ref{alg:heuristic_frequency}. 
Last, in Figure~\ref{fig:final}, we can check that Algorithm~\ref{alg:rank_estimation} is at the same time scalable and accurate, two properties not shared by  the naive sampling method.

\begin{figure}[ht]
\vskip 0.2in
\begin{center}
\centerline{\includegraphics[width=\columnwidth]{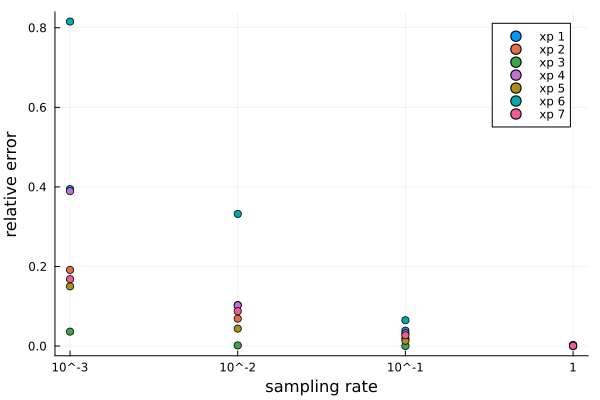}}
\caption{Sampling events before applying directly a sequential simulation might result in very poor estimate. We repeated 7 times the experience for different sampling rate, with an embedding size of 10, 100 campaigns and $10^6$ events. The base budget was 70. 
The error rate is defined as $|\frac{s^{|C|}_N - \hat{s}^{|C|}_N|}{s^{|C|}_N}$.}
\label{fig:icml-sampling}
\end{center}
\vskip -0.2in
\end{figure}

\begin{figure}[ht]
\vskip 0.2in
\begin{center}
\centerline{\includegraphics[width=\columnwidth]{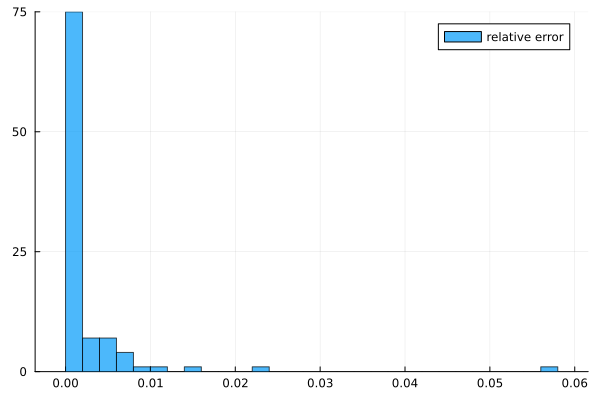}}
\caption{Same setup as in Figure~\ref{fig:icml-sampling}. We compare the output of the sequential and parallel simulations. We can check that they are indeed extremely close. }
\label{fig:icml-parallel}
\end{center}
\vskip -0.2in
\end{figure}

\begin{figure}[ht]
\vskip 0.2in
\begin{center}
\centerline{\includegraphics[width=\columnwidth]{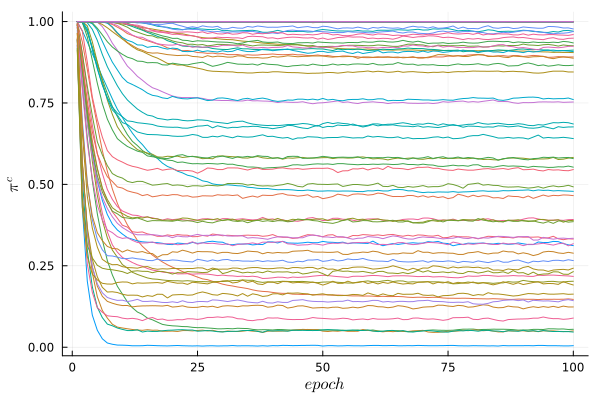}}
\caption{Same setup as in Figure~\ref{fig:icml-sampling}. Convergence of Algorithm~\ref{alg:heuristic_frequency} to estimate the frequencies $\pi^c=\frac{N^c}{N}$. The sampling rate was $0.001$}
\label{fig:frequenciesl}
\end{center}
\vskip -0.2in
\end{figure}

\begin{figure}[ht]
\vskip 0.2in
\begin{center}
\centerline{\includegraphics[width=\columnwidth]{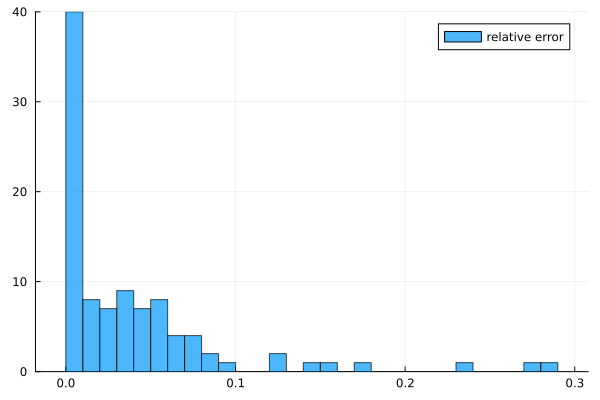}}
\caption{Comparison of the output of \textsc{sort2aggregate} with the ground truth.}
\label{fig:final}
\end{center}
\vskip -0.2in
\end{figure}

\subsection{Yahoo Dataset}
We now move to a more realistic setting by using a dataset released by Yahoo~\footnote{Yahoo! Search Marketing advertiser bidding data, version 1.0, see~\url{http://research.yahoo.com}}.
It is made available to researchers upon request~\cite{yahoo}, while our code will be made available in Python. 
We use one day of advertisers bid on keywords, which allows us to generate a realistic distribution of bids.
There are about 1000 keywords, from which we sample sequences of different sizes. 
We suppose the bids were constant  (we take the average if the advertiser used several bids for a given keyword during this day).

We reproduce a setting were the volume between two consecutive days is increased from 100 000 to 150 000 opportunities, while the budget is fixed (constant budget across all bidders of 2000). 
We use the capping out time of the first day as initialization for the~\textsc{sort2aggregate} algorithm.

In Figure~\ref{fig:yahoo1} we show how the iterates behave on a few selected campaigns. We see that indeed, the campaigns that were spending more than their budget in the first iteration adjust their spend downward. 

We then compare \textsc{sort2aggregate} with two heuristics relying on the first day to predict the second. The first predicts as is, the second applies a rescaling. 
The results are displayed in Figure~\ref{fig:yahoo2}.
We see that \textsc{sort2aggregate} provides a clear advantage over those heuristics. 

\section{Discussion}
Through the abstraction of the auction mechanism, the methodology derived in this work is highly general and adaptable to complex auction settings. It can accommodate scenarios such as multi-slot auctions, which are common in search engine result pages, nested campaigns where multiple line items are linked through a master campaign, and state-dependent cost-per-click (CPC) models that capture diminishing marginal returns as advertising volume increases. This flexibility underscores the broader applicability of our approach beyond simple auction models.  

We believe that such methodologies are crucial in informing decisions that carry significant business, economic, legal, and fairness-related implications. By providing a structured framework for counterfactual estimation in auction environments, our work contributes to the understanding of how change on the platform impact key stakeholders. 

A key challenge remains in extending the methodology to more general functions $ f $, which would, for example, relax the assumption that past auctions influence $ f $ solely through the activation vector. This extension could allow for time-dependent parameters and capture more intricate temporal dependencies in auction dynamics.

\begin{figure}
    \centering
    \includegraphics[width=1.\linewidth]{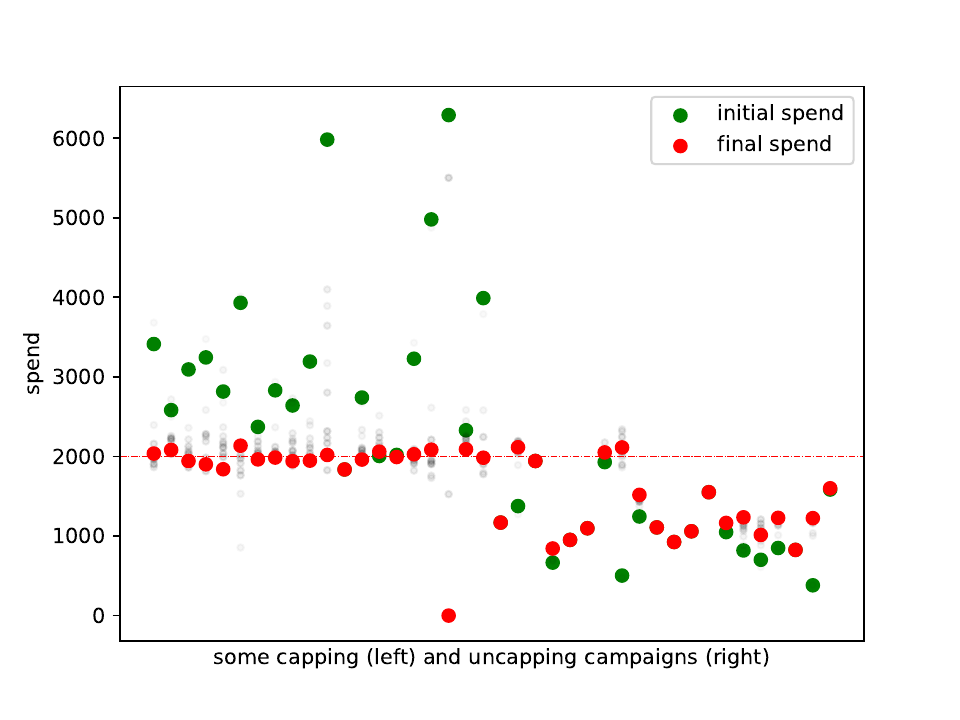}
    \caption{Evolution of the predicted spend for a subset of campaigns over the iterations of the algorithm. The budget (same for all campaigns) is the red dotted horizontal line. }
    \label{fig:yahoo1}
\end{figure}

\begin{figure}
    \centering
    \includegraphics[width=1.\linewidth]{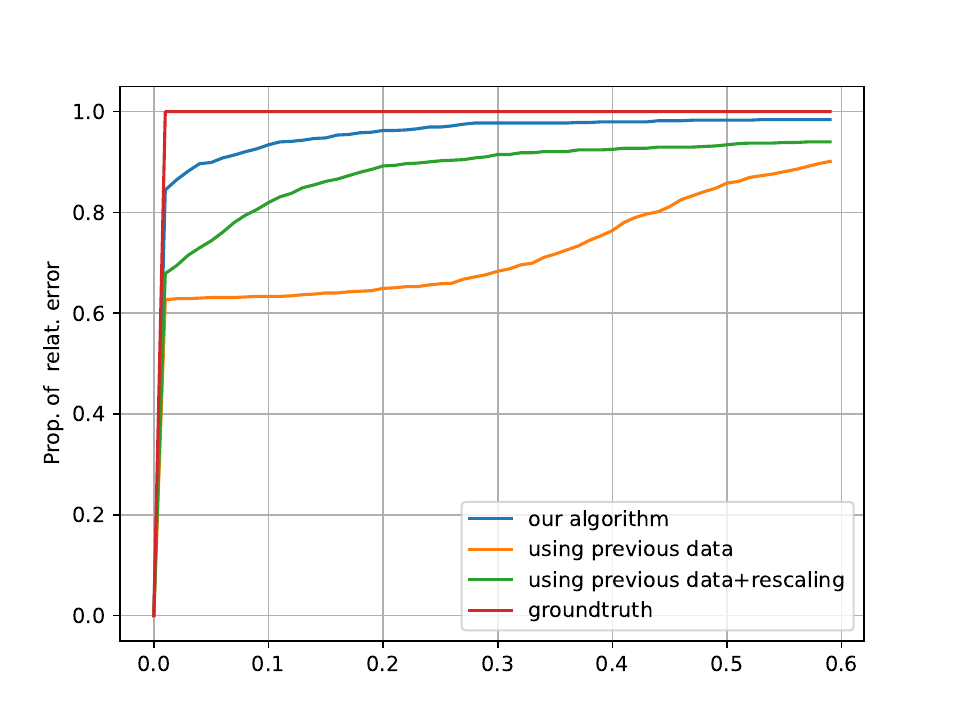}
    \caption{Cumulative of the relative error for the campaigns (weighted by spend). We see that our algorithm improve upon the heuristic of just rescaling the spend by the increase of volume.}
    \label{fig:yahoo2}
\end{figure}

\bibliographystyle{ACM-Reference-Format}
\bibliography{example_paper}

\newpage
\appendix
\onecolumn
\section*{Proofs}
\begin{lemma}
\label{lemma:hoeffding}
Let $n\leqslant N$, $c\in\mathcal{C}$, $\alpha\in\{0,1\}^\mathcal{C}$, and $t>0$, set $F=  \mathbb{E}\bigl[f^c(\boldsymbol{e}_{1}, \alpha) \bigr]$ then 
\begin{align}
    \mathbb{P}\left(\mid\sum_{i=1}^n f^c(\boldsymbol{e}_i,\alpha) - n F\mid \geqslant t \right)\leqslant 2\exp\left(-\frac{2N t^2}{C^2}\right).
\end{align}
\end{lemma}

\begin{proof}
By Hoeffding's inequality~\cite{bardenetConcentrationInequalitiesSampling2015,hoeffding1994probability} and the random order assumption~\ref{assumption:random_order}, and then by the small individual contribution Assumption~\ref{assumption:small-individual-contribution}
\begin{align}
        \mathbb{P}\left(\mid\sum_{i=1}^n f(\boldsymbol{e}_i,\alpha) - n F\mid \geqslant \epsilon n  \right)\leqslant 2\exp\left(-\frac{2n \epsilon^2}{\Vert f\Vert_\infty^2}\right)\leqslant2\exp\left(-\frac{2nN^2 \epsilon^2}{C^2}\right).
\end{align}
Setting $t=\epsilon n $, we get 
\begin{align}
        \mathbb{P}\left(\mid\sum_{i=1}^n f(\boldsymbol{e}_i,\alpha) - n F\mid \geqslant t \right)\leqslant2\exp\left(-\frac{2N^2 t^2}{nC^2}\right)\leqslant2\exp\left(-\frac{2N t^2}{C^2}\right).
\end{align}
\end{proof}

We introduce the times when the $i$th campaign exhausts its budget, formally 
\begin{align}
N_i = \min\{n\in[0,N]:\sum_{c\in\mathcal{C}}(1-a_n^c)\geqslant i\}.
\end{align}

We recall the following definitions from the main text:
\begin{align}
\label{eq:defF}
    F_{i+1} = \mathbb{E}\bigl[f(\boldsymbol{e}_{\hat{N}_i + 1}, \hat{a}_i) \,\big|\,
    \{\hat{N}_i, \boldsymbol{e}_1, \ldots, \boldsymbol{e}_{\hat{N}_i}\}\bigr].
\end{align}
Also, from, 
\begin{align}
 \hat{N}_{i+1} = \min\Bigl(\hat{N}_i \;+\; \bigl\lfloor\frac{b^{\hat{c}_{i+1}} - \hat{s}_{\hat{N}_i}^{\hat{c}_{i+1}}}{F_{i+1}^{\hat{c}_{i+1}}} \bigr\rfloor,\; N\Bigr), 
\end{align}
we deduce that whenever $\hat{N}_{i+1}<N$, 
\begin{align}
 \hat{N}_{i+1} - \hat{N}_i=  \bigl\lfloor\frac{b^{\hat{c}_{i+1}} - \hat{s}_{\hat{N}_i}^{\hat{c}_{i+1}}}{F_{i+1}^{\hat{c}_{i+1}}} \bigr\rfloor 
\end{align}
hence, by definition of $x\to \lfloor x \rfloor$
\begin{align}
\frac{b^{\hat{c}_{i+1}} - \hat{s}_{\hat{N}_i}^{\hat{c}_{i+1}}}{F_{i+1}^{\hat{c}_{i+1}}} - 1 \leqslant\hat{N}_{i+1}- \hat{N}_i \leqslant \frac{b^{\hat{c}_{i+1}} - \hat{s}_{\hat{N}_i}^{\hat{c}_{i+1}}}{F_{i+1}^{\hat{c}_{i+1}}},
\end{align}
hence, since $F^{c_i}_{i+1}>0$
 \begin{align}
b^{\hat{c}_{i+1}} - \hat{s}_{\hat{N}_i}^{\hat{c}_{i+1}} - F_{i+1}^{\hat{c}_{i+1}} \leqslant\left(\hat{N}_{i+1}- \hat{N}_i\right) F_{i+1}^{\hat{c}_{i+1}} \leqslant b^{\hat{c}_{i+1}} - \hat{s}_{\hat{N}_i}^{\hat{c}_{i+1}} 
\end{align}
hence, because by the small individual contribution assumption~\eqref{assumption:small-individual-contribution}, 
$F_{i+1}^{\hat{c}_{i+1}}\leq \frac{C}{N}$,
 \begin{align}
 \label{eq:floor_final}
b^{\hat{c}_{i+1}} - \hat{s}_{\hat{N}_i}^{\hat{c}_{i+1}} - \frac{C}{N}\leqslant\left(\hat{N}_{i+1}- \hat{N}_i\right) F_{i+1}^{\hat{c}_{i+1}} \leqslant b^{\hat{c}_{i+1}} - \hat{s}_{\hat{N}_i}^{\hat{c}_{i+1}} 
\end{align}

\begin{proof}[Main theorem]
Let $c_i$ be the campaign associated with the stopping times $N_i$, we have the following
    \begin{align}
       \eta_{i} \overbrace{:=}^{\text{defined as}}&\mid b^{c_i} - \sum_{n=1}^{\hat{N}_i} f^{c_i}(e_n,a_n) \mid 
        \overbrace{\leqslant}^{\text{triangular inequality}}
        & \underbrace{\mid b^{c_i} - \sum_{n=1}^{\hat{N}_i} f^{c_i}(e_n,\hat{a}_n) \mid}_{A}  + \\ &\underbrace{\mid \sum_{n=1}^{\hat{N}_i} f^{c_i}(e_n,\hat{a}_n) - \sum_{n=1}^{\hat{N}_i} f^{c_i}(e_n,a_n) \mid}_{B}
    \end{align}
    \begin{align*}
        A &= \mid b^{c_i} - \sum_{n=1}^{\hat{N}_i} f^{c_i}(e_n,\hat{a}_n) \mid\\
          &= \mid b^{c_i} \overbrace{- \hat{s}_{\hat{N}_{i-1}}^{c_i} + \underbrace{\hat{s}_{\hat{N}_{i-1}}^{c_i}}_{=\sum_{n=1}^{\hat{N}_{i-1}}f^{c_i}(e_n,\hat{a}_n)}}^{=0} -\sum_{n=1}^{\hat{N}_i} f^{c_i}(e_n,\hat{a}_n) \mid \\
          &= \mid b^{c_i} - \hat{s}_{\hat{N}_{i-1}}^{c_i} +  \sum_{n=1}^{\hat{N}_{i-1}}f^{c_i}(e_n,\hat{a}_n) -\sum_{n=1}^{\hat{N}_i} f^{c_i}(e_n,\hat{a}_n) \mid \\
        &= \mid b^{c_i} - \hat{s}_{\hat{N}_{i-1}}^{c_i} -\sum_{n=\hat{N}_{i-1}+1}^{\hat{N}_i} f^{c_i}(e_n,\hat{a}_n) \mid \\
          &= \mid b^{c_i} - \hat{s}_{\hat{N}_{i-1}}^{c_i} 
         \overbrace{- \left(\hat{N}_i - \hat{N}_{i-1}\right)\mathbb{E}\left[f^{c_i}(e_n,\hat{a}_n)\mid \mathcal{F}_{\hat{N}_{i-1}}\right]
          + \left(\hat{N}_i - \hat{N}_{i-1}\right)\mathbb{E}\left[f^{c_i}(e_n,\hat{a}_n)\mid \mathcal{F}_{\hat{N}_{i-1}}\right]}^{=0}
          \\ & -\sum_{n=\hat{N}_{i-1}+1}^{\hat{N}_i} f^{c_i}(e_n,\hat{a}_n) \mid \\
          &\overbrace{\leqslant}^{\text{triang. ineq.}} \underbrace{\mid b^{c_i} - \hat{s}_{\hat{N}_{i-1}}^{c_i} 
         - \left(\hat{N}_i - \hat{N}_{i-1}\right)\mathbb{E}\left[f^{c_i}(e_n,\hat{a}_n)\mid \mathcal{F}_{\hat{N}_{i-1}}\right]
          \mid}_{\leqslant \frac{C}{N} \text{by~\eqref{eq:floor_final} and~\eqref{eq:defF}}} +\\& \underbrace{\mid \left(\hat{N}_i - \hat{N}_{i-1}\right)\mathbb{E}\left[f^{c_i}(e_n,\hat{a}_n)\mid \mathcal{F}_{\hat{N}_{i-1}}\right]
          -\sum_{n=\hat{N}_{i-1}+1}^{\hat{N}_i} f^{c_i}(e_n,\hat{a}_n) \mid }_{\leqslant t~\text{with probability } 2\exp\left(-\frac{2N t^2}{C^2}\right) \text{by Lemma~\ref{lemma:hoeffding}}}\\
          &\leqslant \frac{C}{N}+t\quad \text{with probability } 1-2\exp\left(-\frac{2N t^2}{C^2}\right)
    \end{align*}

On the other hand, we have
\begin{align*}
    B &=\mid \sum_{n=1}^{\hat{N}_i} f^{c_i}(e_n,\hat{a}_n) - \sum_{n=1}^{\hat{N}_i} f^{c_i}(e_n,a_n) \mid\\
    &\leqslant 
    \mid \underbrace{
    \sum_{j=0}^{i-1}\sum_{n=N_j\wedge\hat{N}_j}^{(N_j\vee \hat{N}_j)\wedge \hat{N}_i}\left(f^{c_i}(e_n,\hat{a}_n) -f^{c_i}(e_n,a_n) \right)}_{\text{controlled by the smooth assumption~\eqref{eq:gamma}}}  + \underbrace{
    \sum_{j=0}^{i-1}\sum_{n=N_j\vee\hat{N}_j +1}^{(N_{j+1}\vee \hat{N}_{j+1})\wedge \hat{N}_i}\left(f^{c_i}(e_n,\hat{a}_n) -f^{c_i}(e_n,a_n) \right) }_{=0~\text{since \ } \hat{a}_n = a_n \text{for $t$ small enough}}
    \mid \\
    &\leqslant 
    \sum_{j=0}^{i-1}\gamma\underbrace{\sum_{n=N_j\wedge\hat{N}_j}^{(N_j\vee \hat{N}_j)\wedge \hat{N}_i}  f^{c_j}(e_n,a_n)}_{\eta_j} +\epsilon \quad \text{with probability $1-\delta$}
\end{align*}
Therefore, with probability at least $1 - \delta- 2\exp\left(-\frac{2N t^2}{C^2}\right)$,
\begin{align}
    \eta_{i}\leqslant \frac{C}{N}+t +\gamma\sum_{j=0}^{i-1}\eta_j+\epsilon.
\end{align}
Remembering the notation $K=|\mathcal{C}|$, we have, by  induction, 
with probability at least $1 - \delta- 2K\exp\left(-\frac{2N t^2}{C^2}\right)$, for all $i\leqslant K $, and for $t$ small enough, 
\begin{align}
    \eta_{i}\leqslant \frac{C}{N}+t+\gamma\sum_{j=0}^{i-1}\left(\eta_j+\epsilon\right).
\end{align}
Therefore (by induction)
\begin{align}
    \eta_i\leqslant u_i
\end{align}
with 
\begin{align}
    u_i = \frac{C}{N}+t+\gamma\sum_{j=0}^{i-1}\left(u_j+\epsilon\right), \quad u_1 = \frac{C}{N}+t+\gamma\epsilon . 
\end{align}
Hence
\begin{align}
    u_i-u_{i-1} = \gamma (u_{i-1}+\epsilon)
\end{align}
therefore
\begin{align}
    u_i  = u_{i-1}(1+\gamma) +\gamma\epsilon,
\end{align}
therefore, recognizing an arithmetico-geomretic sequence, 
\begin{align}
    u_i =(1+\gamma)^{i-1}\left(\frac{C}{N}+t+\gamma\epsilon +\epsilon\right)-\epsilon.
\end{align}
Therefore, with probability at least $1 - \delta- 2K\exp\left(-\frac{2N t^2}{C^2}\right)$,  and for $t$ small enough,
\begin{align}
    \eta_i\leq (1+\gamma)^{i-1}\left(\frac{C}{N}+t+\gamma\epsilon +\epsilon\right)-\epsilon.
\end{align}
We conclude: 
for any $c\in\mathcal{C}$, 
\begin{align}
    \mid s_{N}^c-\hat{s}_{N}^c\mid\leq (1+\gamma)^{K}\left(\frac{C}{N}+t+\gamma\epsilon +\epsilon\right).
\end{align}
\end{proof}
\end{document}